\documentclass[a4paper,UKenglish,cleveref, autoref]{lipics-v2019}

\usepackage{amsmath}
\usepackage{amsfonts}
\usepackage{amssymb}
\usepackage{xspace}
\usepackage[table,xcdraw]{xcolor}
\usepackage{graphicx}
\usepackage{multicol}

\newcommand{\NotNeeded}[1]{}
\newcommand{\FullVersion}[1]{}

\newcommand{\Subject}[1]{\subparagraph{#1.}}
\newcommand{\St}{~|~}

\newcommand{\tuple}[1]{\langle #1  \rangle}
\newcommand{\pair}{\tuple}

\newcommand{\Nat}{\ensuremath{\mathbb{N}}\xspace}

\newcommand{\Dual}[1]{\overline{#1}}

\newcommand{\ExiDerived}[1]{#1\triangledown}
\newcommand{\UniDerived}[1]{#1\!\!\vartriangle}
\newcommand{\Yearn}{\Func{yearn}}

\newcommand{\game}[2]{\ensuremath{\mathcal{G}(#1 , #2)}\xspace} 
\newcommand{\Sgame}[3]{\ensuremath{\mathcal{G}_{#3}(#1 , #2)}\xspace} 
\newcommand{\pro}[2]{\ensuremath{#1 \times #2}\xspace} 

\newcommand{\Func}[1]{{\mathsf{#1}}}
\newcommand{\BP}{\Func{B}^+}

\newcommand{\A}{{\cal A}}
\newcommand{\B}{{\cal B}}
\newcommand{\C}{{\cal C}}
\newcommand{\D}{{\cal D}}

\newcommand{\G}{{\cal G}}
\newcommand{\M}{{\cal M}}

\title{Good for Games Automata:\newline From Nondeterminism to Alternation} 

\titlerunning{Good for Games Automata}

\author{Udi Boker}{Interdisciplinary Center (IDC) Herzliya, Israel}{udiboker@gmail.com}{}{Israel Science Foundation grant 1373/16}

\author{Karoliina Lehtinen}{University of Liverpool, United Kingdom}{k.lehtinen@liverpool.ac.uk}{}{EPSRC grant EP/P020909/1  (Solving Parity Games in Theory and Practice)}

\authorrunning{U. Boker and K. Lehtinen}

\Copyright{Udi Boker and Karoliina Lehtinen}



\category{}

\relatedversion{This is the full version of the paper of the same name presented at the 30th International Conference on Concurrency Theory (CONCUR'19).}

\supplement{}

\acknowledgements{We thank Orna Kupferman for suggesting to look into alternating good-for-games automata and for stimulating discussions on the subject, and Thomas Colcombet for introducing to us his work on alternating history-deterministic automata.}

\nolinenumbers 

\hideLIPIcs  

\EventEditors{John Q. Open and Joan R. Access}
\EventNoEds{2}
\EventLongTitle{42nd Conference on Very Important Topics (CVIT 2016)}
\EventShortTitle{CVIT 2016}
\EventAcronym{CVIT}
\EventYear{2016}
\EventDate{December 24--27, 2016}
\EventLocation{Little Whinging, United Kingdom}
\EventLogo{}
\SeriesVolume{42}
\ArticleNo{23}

\begin{document}

\maketitle

\begin{abstract}
A word automaton recognizing a language $L$ is good for games (GFG) if its composition with any game with winning condition $L$ preserves the game's winner. While all deterministic automata are GFG, some nondeterministic automata are not. There are various other properties that are used in the literature for defining that a nondeterministic automaton is GFG, including ``history-deterministic'', ``compliant with some letter game'', ``good for trees'', and ``good for composition with other automata''. The equivalence of these properties has not been formally shown.

We generalize all of these definitions to alternating automata and show their equivalence.
We further show that alternating GFG automata are as expressive as deterministic automata with the same acceptance conditions and indices. We then show that alternating GFG automata over finite words, and weak automata over infinite words, are not more succinct than deterministic automata, and that determinizing B\"uchi and co-B\"uchi alternating GFG automata involves a $2^{\Theta(n)}$ state blow-up. We leave open the question of whether alternating GFG automata of stronger acceptance conditions allow for doubly-exponential succinctness compared to deterministic automata. 
\end{abstract}

\section{Introduction}
In general, deterministic automata have better compositional properties than nondeterministic automata, making them better suited for applications such as synthesis.
Unfortunately, determinization is complicated and involves an exponential state space blow-up.
Nondeterministic automata that are \textit{good for games} (GFG) have been heralded as a potential way to combine the compositionality of deterministic automata with the conciseness of nondeterministic ones. In this article we are interested in the question of whether the benefits of good-for-games automata extend to alternating automata, which, in general, can be double-exponentially more concise than deterministic automata.

The first hurdle of studying GFG alternating automata is to settle on definitions. Indeed, for nondeterminisic automata, several GFG-type properties have been invented independently under different names: good for games \cite{HP06}, good for trees \cite{KSV06}, and history-determinism \cite{Col09}. 

While Henzinger and Piterman introduced the idea of automata that compose well with games, in their technical development they preferred to use a \textit{letter game} such that one player having a winning strategy in this game implies that the nondeterministic automaton composes well with games \cite{HP06}.
In a similar vein, Kupferman, Safra and Vardi considered already in 1996 a form of nondeterministic automata that resolves its nondeterminism according to the past by looking at tree automata for derived word languages~\cite{KSV06}; this notion of \textit{good for trees} was shown to be equivalent to the letter game~\cite{BKKS13}.
Independently, Colcombet introduced history-determinism in the setting of nondeterministic cost automata \cite{Col09}, and later extended it to alternating automata \cite{Col13}. He showed that history-determinism implies that the automaton is suitable for composition with other alternating automata, a seemingly stronger property than just compositionality with games. Although Colcombet further developed history-determinism for cost automata, here we only consider automata with $\omega$-regular acceptance conditions.

As a result, in the literature there are at least five different definitions that characterize, imply, or are implied by a nondeterministic automaton composing well with games: composition with games, composition with automata, composition with trees, letter games and history-determinism. While some implications between them are proved, others are folklore, or missing. Furthermore, these definitions do not all generalize in the same way to alternating automata: compositionality with games and with automata are agnostic to whether the automaton is nondeterministic, universal or alternating, and hence generalize effortlessly to alternating automata; the letter-game and good-for-tree automata on the other hand generalize `naturally' in a way that treats nondeterminism and universality asymmetrically and hence need be adapted to handle alternation.

In the first part of this article, we give a coherent account of good-for-gameness for alternating automata: we generalize all the existing definitions from nondeterministic to alternating automata, and show them be equivalent. This implies that these are also equivalent for nondeterministic automata. While some of these equivalences were already folklore for nondeterministic automata, others are more surprising: compositionality with one-player games implies compositionality with two-player games and compositionality with automata, despite games being a special case of alternating automata and single-player games being a special case of games. We also show that in the nondeterministic case each definition can be relaxed to an asymmetric requirement: composition with universal automata  or universal trees is already equivalent to composition with alternating automata and games.

In the second part of this article, we focus on questions of expressiveness and succinctness. 
The first examples of GFG automata were built on top of deterministic automata \cite{HP06}, and Colcombet conjectured that history-deterministic alternating automata with $\omega$-regular acceptance conditions are not more concise than deterministic ones \cite{Col12}.
Yet, this has since been shown to be false: already GFG nondeterministic B\"uchi automata cannot be pruned into deterministic ones \cite{BKKS13} and co-B\"uchi automata can be exponentially more concise than deterministic ones \cite{KS15}. In general, nondeterministic GFG automata are in between nondeterministic and deterministic automata, having some properties from each \cite{BKS17}.  

Whether GFG alternating automata can be double exponentially more concise than deterministic automata is  particularly interesting in the wake of quasi-polynomial algorithms for parity games.
Indeed, since 2017 when Calude et al. brought down the upper bound  for solving parity games from subexponential to quasi-polynomial~\cite{CJKLS17}, the automata-theoretical aspects of solving parity games with quasi-polynomial complexity have been studied in more depth~\cite{Leh18,BL19,BL18,DJL19,BC18,CDFJLP19,CF19}. In particular,
 Boja\'nczyk and Czerwi\'nski \cite{BC18}, and Czerwi{\'n}ski et al.~\cite{CDFJLP19} describe the quasi-polynomial algorithms for solving parity games explicitly in terms of deterministic word automata that separate some word languages. A polynomial deterministic \textit{or GFG} safety separating automaton for these languages would imply a polynomial algorithm for parity games. However, it is shown in \cite{CDFJLP19} that the smallest possible such \textit{nondeterministic} automaton is quasi-polynomial. As this lower bound only applies for nondeterministic automata, it is interesting to understand whether alternating GFG automata could be more concise.

As for expressiveness, we show that alternating GFG automata are as expressive as deterministic automata with the same acceptance conditions and indices. The proof extends the technique used in the nondeterministic setting \cite{BKKS13}, producing a deterministic automaton by taking the product of the automaton and the two transducers that model its history-determinism.

Regarding succinctness, we first show that GFG automata over finite words, as well as weak automata over infinite words, are not more succinct than deterministic automata. The proof builds on the property that minimal deterministic automata of these types have exactly one state for each Myhill-Nerode equivalence class, and an analysis that GFG automata of these types must also have at least one state for each such class.  

We proceed to show that determinizing B\"uchi and co-B\"uchi alternating GFG automata involves a $2^{\theta(n)}$ state blow-up. 
The proof in this case is more involved, going through two main lemmas. The first shows that for alternating GFG B\"uchi automata, a history-deterministic strategy need not remember the entire history of the transition conditions, and can do with only remembering the prefix of the word read. The second lemma shows that the breakpoint (Miyano-Hayashi) construction, which is used to translate an alternating B\"uchi automaton into a nondeterministic one, preserves GFGness.
We leave open the question of whether alternating GFG automata of stronger acceptance conditions allow for doubly-exponential succinctness compared to deterministic automata. 

\section{Preliminaries}\label{sec:Preliminaries}

\Subject{Words and automata}
An \emph{alphabet} $\Sigma$ is a finite nonempty set of letters, a finite (resp.\ infinite) \emph{word} $u=u_0 \ldots u_k\in \Sigma^{*}$ (resp.\ $w=w_0 w_1\ldots\in \Sigma^{\omega}$) is a finite (resp.\ infinite) sequence of letters from $\Sigma$. 
A \emph{language} is a set of words, and the empty word is written $\epsilon$.
 
An \emph{alternating word automaton} is $\A=(\Sigma,Q,\iota,\delta,\alpha)$, where $\Sigma$ is a finite nonempty alphabet, $Q$ is a finite nonempty set of states, $\iota\in Q$ is an initial state, $\delta:Q\times \Sigma \rightarrow \BP(Q)$ is a transition function where $\BP(Q)$ is the set of positive boolean formulas (\emph{transition conditions}) over $Q$, and $\alpha$ is an acceptance condition, on which we elaborate below.
For a state $q\in Q$, we denote by $\A^q$ the automaton that is derived from $\A$ by setting its initial state to $q$. 
$\A$ is nondeterministic (resp.\ universal) if all its transition conditions are disjunctions (resp.\ conjunctions), and it is deterministic if all its transition conditions are states. 

There are various acceptance conditions, defined with respect to the set of states that (a path of) a run of $\A$ visits. Some of the acceptance conditions are defined on top of a labeling of $\A$'s states. In particular, the parity condition is a labeling $\alpha:Q \to \Gamma$, where $\Gamma$ is a finite set of priorities and a path is accepting if and only if the highest priority seen infinitely often on it is even.
A B\"uchi condition is the special case of the parity condition where $\Gamma=\{1,2\}$; states of priority $2$ are called \textit{accepting} and of priority $1$ \textit{rejecting}, and then $\alpha$ can be viewed as the subset of accepting states of $Q$. Co-B\"uchi automata are dual, with $\Gamma=\{0,1\}$. A weak automaton is a B\"uchi automaton in which every strongly connected component of the graph induced by the transition function consists of only accepting or only rejecting states.

In Sections~\ref{sec:Preliminaries}-\ref{sec:Expressiveness}, we handle automata with arbitrary $\omega$-regular acceptance conditions, and thus consider $\alpha$ to be a mapping from $Q$ to a finite set $\Gamma$, on top of which some further acceptance criterion is implicitly considered (as in the parity condition). In Section~\ref{sec:Succinctness}, we focus on weak, B\"uchi, and co-B\"uchi automata, and then view $\alpha$ as a subset of $Q$. 

\newcommand{\GL}{C} 
\Subject{Games}
A \emph{finite $\Sigma$-arena} is a finite $\Sigma\times \{A,E\}$-labeled Kripke structure. An \emph{infinite $\Sigma$-arena} is an infinite $\Sigma\times \{A,E\}$-labeled tree. Nodes with an $A$-label are said to belong to Adam; those with an $E$-label are said to belong to Eve. We represent a  $\Sigma$-arena as $R=(V,X,V_E,\GL)$, where $V$ is its set of nodes, $X\subseteq V\times V$ its transitions, $V_E\subseteq V$ the $E$-labeled nodes, $V\setminus V_E$ the $A$-labeled nodes and $\GL:V\to\Sigma$ its $\Sigma$-labeling function. We will assume that all states have a successor.
An arena might be rooted at an initial position $v_\iota\in V$.

A play in $R$ is an infinite path in $R$. 
A \emph{game} is a $\Sigma$-arena together with a winning condition $W\subseteq \Sigma^\omega$. A play $\pi$ is said to be winning for Eve in 
the game if the $\Sigma$-labels along $\pi$ form a word in $W$. Else $\pi$ is winning for Adam.

A \emph{strategy} for Eve (Adam, resp.) is a function $\tau:V^*\rightarrow V$ that maps a history $v_0\ldots v_i$, i.e. a finite prefix of a play in $R$, to a successor of $v_i$ whenever $v_i\in V_E$ ($v_i\notin V_E$). A play $v_0v_1\dots$ agrees with a strategy $\tau$ for Eve (Adam) if whenever $v_{i}\in V_E$ ($v_i\notin V_E$), we have $v_{i+1}=\tau(v_i)$. A strategy for Eve (Adam) is winning from a position $v\in V$ if all plays beginning in $v$ that agree with it are winning for Eve (Adam). We say that a player wins the game from a position $v\in V$ if they have a winning strategy from $v$. 
If the game is rooted at $v_\iota$, we say that a player wins the game if they win from $v_\iota$.

All the games we consider have $\omega$-regular winning conditions and are therefore determined and the winner has a finite-memory strategy~\cite{BL69}.
Finite-memory strategies can be modeled by {\em transducers}. Given alphabets $I$ and $O$, an \emph{$(I/O)$-transducer} is a tuple ${\M}=(I,O,M,\iota,\rho,\chi)$, where  $M$ is a finite set of states (memories), $\iota\in M$ is an initial memory, $\rho:M\times I \to M$ is a deterministic transition function, and $\chi:M\to O$ is an output function. The strategy $\M:I^*\to O$ is  obtained by following $\rho$ and $\chi$ in the expected way: we first extend $\rho$ to words in $I^\ast$ by setting $\rho(\epsilon)=\iota$ and $\rho(u \cdot a)=\rho(\rho(u),a)$, and then define $\M(u)=\chi(\rho(u))$.

\Subject{Products}
\begin{definition}[Synchronized product]\label{def:GameAutomataProd}
The \emph{synchronized product} $\pro{R}{\A}$ between a $\Sigma$-arena $R=(V,X,V_E, \C)$ and an alternating automaton $\A=(\Sigma,Q,\iota,\delta,\alpha)$ with mapping $\alpha:Q \to \Gamma$ is a $\Gamma \cup \{\bot\}$-arena of which the states are $\pro{V}{\BP(Q)}$ and the successor relation is defined by:
\begin{itemize}
\item $(v,q)$, for a state $q$ of $Q$, has successors $(v',\delta(q,\C(v')))$ for each successor $v'$ of $v$ in $R$;
\item $(v,b\wedge b')$ and $(v,b\vee b')$ have two successors $(v,b)$ and $(v,b')$;
\item If $R$ is rooted at $v$ then the root of $R\times \A$ is $(v, \delta(\iota,\C(v)))$.
\end{itemize}
The positions belonging to Eve are $(v,b)$ where either $b$ is a disjunction, or $b$ is a state in $Q$ and $v\in V_E$.
The label of $(v,b)$ is $\alpha(b)$ if $b$ is a state of $Q$, and $\bot$ otherwise.
\end{definition}

An example, without labeling, of a synchronized product is given in Figure~\ref{fig:Product}.
\begin{figure}
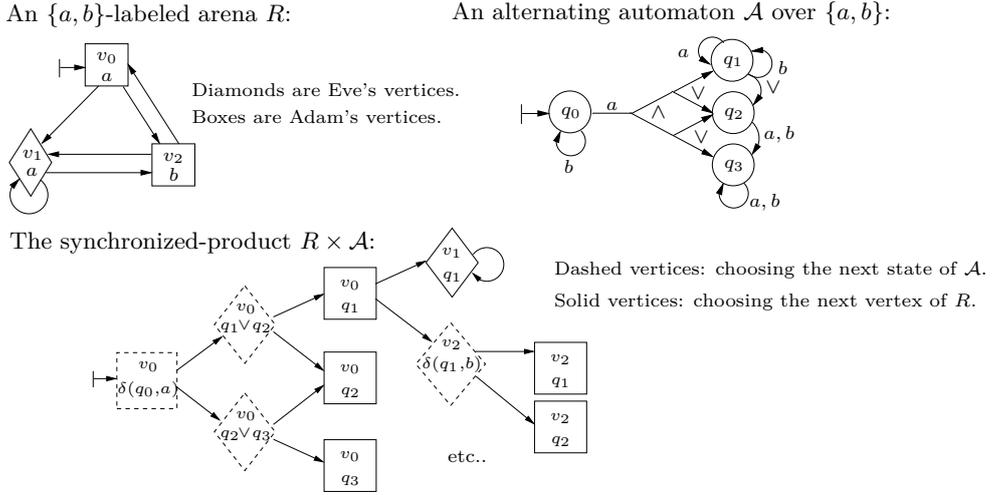 \caption{An example of a product between an alternating automaton and a finite arena.}\label{fig:Product}
\end{figure}
\begin{definition}[Automata composition]\label{def:prod}

Given alternating automata $\B=(\Sigma,Q^\B, \iota^\B,\delta^\B,\beta: Q^\B \rightarrow \Gamma)$ and $\A=(\Gamma, Q^\A,\iota^\A,\delta^\A,\alpha)$, their \emph{composition} $\pro{\B}{\A}$ consists of the synchronized product automaton $(\Sigma,Q^\B\times Q^\A,(\iota^\B,\iota^\A),\delta,\alpha')$, where $\alpha'(q^B,q^A)=\alpha(q^A)$ and $\delta((q^\B, q^\A), a)$ consists of  $f(\delta^\B(q^B,a),q^\A)$ where:

\begin{multicols}{2}
\begin{itemize}
\item $f(c\vee c',q)=f(c,q)\vee f(c',q)$
\item $f(c\wedge c',q)=f(c,q) \wedge f(c',q)$
\item $f(q',q)=g(q',\delta^\A(q,\beta(q'))$ where
\item $g(q,c\vee c')=g(q,c)\vee g(q,c')$
\item $g(q,c\wedge c')=g(q,c)\wedge g(q,c')$
\item $g(q,q')=(q,q')$.
\end{itemize}
\end{multicols}

\end{definition}
Note that this stands for first unfolding the transition condition in $\B$ and then the transition condition in $\A$, and it is equivalent to the following substitution, which matches Colcombet's notation \cite{Col13}:
$\delta^\B(q^\B,a)[q\in Q^\B \leftarrow \delta^\A(q^\A, \beta(q))[p\in Q^\A \leftarrow (q,p)]]$.
\Subject{Acceptance of a word by an automaton} We define the acceptance directly in terms of the model-checking (membership) game, which happens to be exactly the product of the automaton with a path-like arena describing the input word. More precisely, $\A$ accepts a word $w$ if and only if Eve wins the \emph{model-checking game} $\game{w}{\A}$, defined as the product $\pro{R_w}{\A}$, where the arena $R_w$ consists of an infinite path, of which all positions belong to Eve (although it does not matter), and the label of the $i^{\textit{th}}$ position is the $i^{\textit{th}}$ letter of $w$. The language of an automaton $\A$, denoted by $L(\A)$, is the set of words that it accepts (recognizes). Two automata are equivalent if they recognize the same language.
We will refer to the positions of the arena $R_w$ by the suffix of $w$ that labels the path starting there. In particular, this gives us a finite representation of the arena for ultimately periodic words. We further denote by $\Sgame{w}{\A}{\tau}$ the model-checking game $\game{w}{\A}$ restricted to moves that agree with the strategy $\tau$ of Adam or Eve.

\section{Good for Games Automata: Five Definitions}\label{sec:GfgDefinitions}
We clarify the definitions that are used in the literature for ``good for gameness'', while generalizing them from a nondeterministic to an alternating word automaton $\A=(\Sigma,Q,\iota,\delta,\alpha)$.

\Subject{Good for game composition}
The first definition matches the intuition that ``$\A$ is good for playing games''. It was given in \cite{HP06} for nondeterministic automata and applies as is to alternating automata, by properly defining the synchronized product of a game and an alternating automaton. (Definition~\ref{def:GameAutomataProd}.)
We prove in Section~\ref{sec:Equivalence} that it is equivalent when speaking of only one-player finite-arena games and two-player finite/infinite-arena games.

\begin{definition}[Good for game composition]\label{def:GFGameComposition}
$\A$ is \emph{good for game composition} if for every [one-player] game $G$ with a [finite] $\Sigma$-labeled arena and winning condition $L(\A)$, Eve has a winning strategy in $G$ if and only if she has a winning strategy in the synchronized-product game $\pro{G}{\A}$.
\end{definition}

\Subject{Compliant with the letter games}
The first definition is simple to declare, but is not convenient for technical developments. Thus, Henzinger and Piterman defined the ``letter-game'', our next definition, while independently Colcombet defined history-determinism, which we provide afterwards. 
The two latter definitions are easily seen to be equivalent and they were shown in~\cite{HP06} to imply the game composition definition. We are not aware of a full proof of the other direction in the literature; we include one in this article.

In the letter-game for nondeterministic automata \cite{HP06}, Adam generates a word letter by letter, and Eve has to resolve the nondeterminism ``on the fly'', so that the generated run of $\A$ is accepting if Adam generates a word in the language.
It has not been generalized yet to the alternating setting, and there are various ways in which it can be generalized, as it is not clear who should pick the letters and how to resolve the nondeterminism and universality. It turns out that a generalization that works well is to consider two independent games, one in which Eve resolves the nondeterminism while Adam picks the letters and resolves the universality, and another in which Eve picks the letters. 
\begin{definition}[Compliant with the letter games]\label{def:GFLetterGame}
There are two letter games, Eve's game and Adam's game.
\begin{description}
\item[Eve's game:] 
A configuration is a pair $(\sigma, b)$ where $b\in \BP(Q)$ is a transition condition and $\sigma\in\Sigma\cup\{\epsilon\}$ is a letter. (We abuse $\epsilon$ to also be an empty letter.)
A play begins in $(\sigma_0,b_0)=(\epsilon,\iota)$ and consists of an infinite sequence of configurations $(\sigma_0,b_0)(\sigma_1,b_1)\ldots$. 
In a configuration $(\sigma_i,b_i)$, the game proceeds to the next configuration $(\sigma_{i+1},b_{i+1})$ as follows.  
\begin{itemize}
\item If $b_i$ is a state of $Q$, Adam picks a letter $a$ from $\Sigma$, having $(\sigma_{i+1},b_{i+1})=(a,\delta(b_i,a))$.
\item If $b_i$ is a conjunction $b_i=b' \land b''$, Adam chooses between $(\epsilon,b')$ and $(\epsilon,b')$.
\item If $b_i$ is a disjunction $b_i=b' \lor b''$, Eve chooses between $(\epsilon,b')$ and $(\epsilon,b')$.
\end{itemize}
In the limit, a play consists of an infinite sequence $\pi = b_0,b_1,\ldots$ of transition conditions and an infinite word $w$ consisting of the concatenation of $\sigma_0,\sigma_1,\ldots$. Let $\rho$ be the restriction of $\pi$ to transition conditions that are states of $Q$. Eve wins the play if either $w\notin L(\A)$ or $\rho$ satisfies $\A$'s acceptance condition.

The \emph{nondeterminism in $\A$ is compliant with the letter games} if Eve wins this game.

\item[Adam's game] is similar to Eve's game, except that Eve chooses the letters instead of Adam, and Adam wins if either $w\in L(\A)$ or $\rho$ does not satisfy $\A$'s acceptance condition. The \emph{universality in $\A$ is compliant with the letter games} if Adam wins this game.
\end{description}
$\A$ is \emph{compliant with the letter games} if its nondeterminism and universality are compliant with the letter games.

\end{definition}
 Observe that we need both games: universal automata are trivially compliant with Eve's letter game and  nondeterministic automata are trivially compliant with Adam's letter game.

\Subject{History-determinism}
A nondeterministic automaton is history-deterministic \cite{Col09} if there is a strategy\footnote{In Section~\ref{sec:Preliminaries}, we formally defined a ``strategy'' with respect to a specific game. Here we abuse the term ``strategy'' to refer to a general total function on finite words.} to resolve the nondeterminism that only depends on the word read so far, i.e., that is uniform for all possible futures. Colcombet generalized this to alternating automata \cite{Col13} by considering such strategies for both players.

We first define how to use a strategy $\tau: (\pro{\Sigma}{\BP(Q)})^*\rightarrow \BP(Q)$ for playing in a model-checking game $\game{w}{\A}$, as the history domains are different. Recall that in the model-checking game $\game{w}{\A}$, positions consist of  a suffix of $w$ and a transition condition of $\A$ so histories have type $(\pro{\Sigma^{\omega}}{\BP(Q)})^*$. From such a history $h$, let $h'$ be the history obtained by only keeping the first letter of the $\Sigma^{\omega}$ component of $h$'s elements, that is, the letter at the head of the current suffix. Then, we extend $\tau$ to operate over the $(\pro{\Sigma^{\omega}}{\BP(Q)})^*$ domain, by defining $\tau(h)=\tau(h')$.

For convenience, we often refer to a history in  $(\pro{\Sigma}{\BP(Q)})^*$, as a pair in $\pro{\Sigma^*}{(\BP(Q))^*}$.

\begin{definition}[History-determinism \cite{Col13}]\label{def:HistoryDet}\
\begin{itemize}
\item The \emph{nondeterminism in $\A$ is history-deterministic} if there is a strategy $\tau_E: (\pro{\Sigma}{\BP(Q)})^*\rightarrow \BP(Q)$ such that for all $w\in L(\A)$, $\tau_E$ is a winning strategy for Eve in $\game{w}{\A}$.  
\item The \emph{universality in $\A$ is history-deterministic} if there is a strategy $\tau_A: (\pro{\Sigma}{\BP(Q)})^*\rightarrow \BP(Q)$ such that for all $w\notin L(\A)$, $\tau_A$ is a winning strategy for Adam in $\game{w}{\A}$.
\item $\A$ is \emph{history-deterministic} if its nondeterminism and universality are history-deterministic.
\end{itemize}
\end{definition}

\Subject{Good for automata composition}
The next definition comes from Colcombet's proof that alternating history-deterministic automata behave well with respect to composition with other alternating automata. 
We shall show in Section~\ref{sec:Equivalence} that it also implies proper compositionality with tree automata, and that for nondeterministic automata, it is enough to require compositionality with universal, rather than alternating, automata.

\begin{definition}[Good for automata composition \cite{Col13}]\label{def:GFAutomataComposition}
$\A$ is \emph{good for automata composition} if for every alternating word (or tree) automaton $\B$ with $\Sigma$-labeled states and acceptance condition $L(\A)$, the language of the composed automaton $\pro{\B}{\A}$ is equal to the language of $\B$.
\end{definition}

\Subject{Good for trees}
The next definition comes from the work in \cite{KSV06,BKKS13} on the power of nondeterminism in tree automata. It states that a nondeterministic word automaton $\A$ is good-for-trees if we can ``universally expand'' it to run on trees and accept the ``universally derived language'' $\UniDerived{L(\A)}$---trees all of whose branches are in the word language of $\A$. 

Observe that every universal word automaton for a language $L$ is trivially good for $\UniDerived{L}$. Therefore,
for universal automata, we suggest that a dual definition is more interesting: its ``existential expansion to trees'' accepts $\ExiDerived{L}$---trees in which there exists a path in $L$. 

For an alternating automaton $\A$, we generalize the good-for-trees notion to require that $\A$ is good for both $\UniDerived{L(\A)}$ and $\ExiDerived{L(\A)}$, when expanded universally and existentially, respectively.

We first formally generalize the definition of ``expansion to trees'' to alternating automata. 
The \emph{universal (resp.\ existential) expansion} of $\A$ to trees is syntactically identical to $\A$, while its semantics is to accept a tree $t$ if and only if Eve wins the game $t \times \A$, when $t$ is viewed as a game in which all nodes belong to Adam (resp.\ Eve).

\begin{definition}[Good for trees]\label{def:GFTrees}
$\A$ is \emph{good for trees} if its universal- and existential-expansions to trees recognize the tree languages $\UniDerived{L(A)}$ and $\ExiDerived{L(A)}$, respectively.
\end{definition}

\section{Equivalence of All Definitions}\label{sec:Equivalence}

We prove in this section the equivalence of all of the definitions in the alternating setting, as given in Section~\ref{sec:GfgDefinitions}, which implies their equivalence also in the nondeterministic (and universal) setting. We may therefore refer to an automaton as \emph{good-for-games (GFG)} if it satisfies any of these definitions. In some cases, we provide additional equivalences that only apply to the nondeterministic setting.

\begin{theorem}\label{thm:equivalence}
Given an alternating automaton $\A$, the following are equivalent:
\begin{enumerate}
\item $\A$ is good for game composition (Def. \ref{def:GFGameComposition}).
\item $\A$ is compliant with the letter games (Def. \ref{def:GFLetterGame}).
\item $\A$ is history-deterministic (Def. \ref{def:HistoryDet}).
\item $\A$ is good for automata composition (Def. \ref{def:GFAutomataComposition}).
\item $\A$ is good for trees (Def. \ref{def:GFTrees}).
\end{enumerate}
\end{theorem}

\begin{proof}\
\begin{itemize}
\item Lemma~\ref{lem:letter-hd}: History-determinism $\Leftrightarrow$ compliance with the letter games. (Def.~\ref{def:GFLetterGame} $\Leftrightarrow$ Def.~\ref{def:HistoryDet}).

\item Lemma~\ref{lem:LetterImpliesGfg}: Compliance with the letter games $\Rightarrow$ compositionality with arbitrary games. (Def.~\ref{def:GFLetterGame} $\Rightarrow$ ``strong'' Def.~\ref{def:GFGameComposition}). 

\item Lemma~\ref{lem:GFfiniteImpliesLetter}: Compositionality, even with just one-player finite-arena games $\Rightarrow$ compliance with the letter games. (``weak'' Def.~\ref{def:GFGameComposition} $\Rightarrow$ Def.~\ref{def:GFLetterGame}). 

\item Lemma~\ref{lem:GoodForTrees}: Good for trees is $\Leftrightarrow$ compositionality with one-player games.\\(Def.~\ref{def:GFTrees} $\Leftrightarrow$ ``medium'' Def.~\ref{def:GFGameComposition}).

\item Lemma~\ref{lem:AutomataComposition}: Compositionality with games $\Leftrightarrow$  compositionality with automata.\\(Def.~\ref{def:GFAutomataComposition} $\Leftrightarrow$ Def.~\ref{def:GFGameComposition}).
\end{itemize}
\end{proof}

We start with the simple equivalence of history-determinism and compliance with the letter game. (Observe that the letter-game strategies are of the same type as the strategies that witness history-determinism: a function from $(\pro{\Sigma}{\BP(Q)})^*$ to $\BP(Q)$.)
\begin{lemma}\label{lem:letter-hd}
Consider an alternating automaton $\A=(\Sigma,Q,\iota,\delta,\alpha)$.
\begin{itemize}
\item A strategy $\tau_E$ for Eve in her letter game is winning if and only if it witnesses the history-determinism of the nondeterminism in $\A$.
\item  A strategy $\tau_A$ for Adam in his letter game is winning if and only if it witnesses the history-determinism of the universality in $\A$. 
\item An alternating automaton $\A$ is history-deterministic if and only if it is compliant with the letter games.
\end{itemize}
\end{lemma}

\begin{proof}
For the first direction, we assume that the nondeterminism in $\A$ is history-deterministic, witnessed by a strategy $\tau_E$ of Eve.
Then Eve wins her letter game by following $\tau_E$, since if Adam plays a word $w\in L(\A)$, then the resulting play of the letter game, consisting of a sequence $\pi=b_0,b_1\ldots$ of transition conditions and a word $w=w_0,w_1\ldots$, induces a play in $\game{w}{\A}$ that agrees with $\tau_E$. Since $\tau_E$ witnesses the history-determinism of $\A$, such a play must be winning, that is, $\pi$ restricted to $Q$ must satisfy $\A$'s acceptance condition.
 
Symmetrically, if the universality in $\A$ is history-deterministic, witnessed by a strategy $\tau_A$ of Adam, it induces a winning strategy for him in his letter-game.
 
For the converse, assume that Eve wins her letter game with a strategy $s$. We argue that this strategy also witnesses the history-determinism of the nondeterminism in $\A$, namely that Eve wins $\Sgame{w}{\A}{s}$ for every word $w\in L(\A)$.

Indeed, if a play $\pi\in\Sgame{w}{\A}{s}$ does not satisfy the acceptance condition of $\A$ while $w\in L(\A)$, then the play in Eve's letter game in which Adam plays $w$ and resolves universality as in $\pi$ would both agree with $s$ and be winning for Adam, contradicting that $s$ is winning for Eve. The nondeterminism of $\A$ is therefore history-deterministic.
 
Symmetrically, if Adam wins his letter game with strategy $\tau_A$, the universality in $\A$ is history-deterministic, witnessed by $\tau_A$. Hence, if $\A$ is compliant with the letter games, it is also history-deterministic.
\end{proof}
\begin{corollary}\label{cor:FiniteMemory}
If $\A$ is history-deterministic, then there are finite-memory strategies $\tau_E$ and $\tau_A$ to witness it.
\end{corollary}

\begin{proof}
Since the letter game is a finite $\omega$-regular game, its winner has a finite-memory strategy.
\end{proof}

The following two propositions state that ``standard manipulations'' of alternating automata preserve history-determinism. The \emph{dual} of an automaton $\A$, denoted by $\Dual{\A}$, is derived from $\A$ by changing every conjunction of a transition condition to a disjunction, and vice versa, and changing the acceptance condition to reject every sequence of states (labeling) that $\A$ accepts, and accept every sequence that $\A$ rejects.

\begin{proposition}\label{prop:DualPreservesGFG}
Consider an alternating automaton $\A$ and its dual $\Dual{\A}$. 
The nondeterminism (resp.\ universality) of $\A$ is history-deterministic if and only if the universality (resp.\ nondeterminism) of $\Dual{\A}$ is history-deterministic.
\end{proposition}

\begin{proof}
For every word $w$, the model-checking games $\game{w}{\A}$ and $\game{w}{\Dual{\A}}$ are the same, just switching roles between Adam and Eve. Thus, the history-deterministic strategy for Adam in $\A$ can serve Eve in $\Dual{\A}$ and vice versa.
\end{proof}

\begin{proposition}\label{prop:TransitionConditionFormat}
Consider an alternating automaton $\A$, and let $\A'$ be an automaton that is derived from $\A$ by changing some transition conditions to different, but equivalent, boolean formulas. Then the nondeterminism/universality in $\A=(\Sigma,Q,\iota,\delta,\alpha)$ is history-deterministic if and only if it is history-deterministic in $\A'$.
\end{proposition}

\begin{proof}
It is enough to show that changing any transition condition of an alternating automaton $\A$ to DNF does not influence its history-determinism for Eve/Adam.

Assume that the nondeterminism in $\A$ is history-deterministic, witnessed by a strategy $\tau$ of Eve. Let $\A'$ be an automaton that is derived from $\A$ by changing any transition condition $b$ of $A$, for a state $q$ and a letter $a$, into its DNF form $b'$. 
Let $k\in\Nat$ be the depth of alternation between nondeterminism and universality in $b$.

We show that Eve has a history-deterministic strategy $\tau'$ for $\A'$, by adapting $\tau$. We call \textit{local $b$-strategy} a way of resolving the nondeterminism in $b$. First observe that for every local $b$-strategy $s$, there is a corresponding local $b'$-strategy $s'$ that chooses the set of states that Adam can force in $k$ steps over $b$ if Eve follows $s$; conversely for every local $b'$-strategy $s'$, there is a corresponding local $b$-strategy $s$ such that the set of states that Adam can force in $k$ steps over $b$ if Eve follows $s$ is exactly Eve's choice in $s'$.

The strategy $\tau'$ can then be defined by replacing $b$-local strategies from $\tau$ with the corresponding $b'$-local strategies. More precisely, for every history $(u,h')\in(\pro{\Sigma^*}{\BP(Q))^*}$, we have that $\tau'(u,h'q)$ is the $b'$-local strategy corresponding to $\tau(u,hq)$, where $h$ is the sequence of transition conditions derived from $h'$, by replacing the $b'$ transitions consistent with a $b'$-local strategy $s'$ with the $b$ transitions consistent with the corresponding $b$-local strategy $s$.
Since the corresponding local strategies only differ in the paths taken within $b$ and $b'$, but not in the resulting states reached, $\tau'$ preserves Eve's victory.

The arguments for the other direction, that is assuming that the nondeterminism in $\A'$ is history-deterministic, and proving that this is also the case for $\A$, are analogous, and so are the arguments for how to adapt a history-deterministic strategy for Adam.
\end{proof}

The following lemma was shown in \cite{HP06} for nondeterministic automata and can be deduced for alternating automata from Lemma~\ref{lem:letter-hd} and the equivalence of history-determinism and being good for composition with automata \cite{Col13}. We provide a simple direct proof.

\begin{lemma}\label{lem:LetterImpliesGfg}
If an alternating automaton $\A$ is compliant with the letter games then it is good for game composition.
\end{lemma}

\begin{proof}
Assume $\A$ is compliant with the letter games but that for some $\Sigma$ arena $G$, the game on $G$ with winning condition $L(\A)$ and the synchronized product $\pro{G}{\A}$ have different winners. If Eve wins in $G$, then she can combine her winning strategy $\tau$ in $G$ and her winning strategy $\tau'$ in her letter-game to win in the synchronized product $\pro{G}{\A}$: she resolves the choices in $G$ according to $\tau$, thus ensuring that the play in $\pro{G}{\A}$ follows a path of $G$ labeled with a word accepted by $\A$. Then, she can resolve the nondeterminism in $\A$ according to $\tau'$. Since $\tau'$ is winning in the letter game and all plays agreeing with $\tau$ follow a word in $G$ that is in $L(A)$, all plays agreeing with the combination of $\tau$ and $\tau'$ are accepting.

Similarly, if Adam wins in $G$, his strategy in $G$ and in his letter game combine into a winning strategy in the product $\pro{G}{\A}$.
\end{proof}

If $\A$ is good for infinite games, it is clearly good for finite games, which can be unfolded into infinite games. The following lemma shows that the other direction holds too: compositionality with finite games implies compliance with the letter games, and therefore, from the previous lemma, composition with infinite games.

Note that this correspondence does not extend to the notion of \textit{good for small games} \cite{BL19,CF19}: an automaton can be good for composition with games up to a bounded size, without being good for games.

\begin{lemma}\label{lem:GFfiniteImpliesLetter}
If an alternating automaton is good for composition with finite-arena one-player games then it is compliant with the letter games.
\end{lemma}
\begin{proof}
Consider an alternating automaton $\A$ over the alphabet $\Sigma$. We show that if $\A$ is not compliant with the letter games then it is not good for finite-arena one-player games. 
By Proposition~\ref{prop:TransitionConditionFormat}, we may assume that the transition conditions in $\A$ are in CNF.

If $\A$ is not compliant with Eve's letter game, then since this game is $\omega$-regular, Adam has some finite-memory winning strategy, modeled by a transducer $\M$. 
Observe that states of $\M$ output the next moves of Adam, namely a letter and a disjunctive clause, and the transitions of $\M$ correspond to moves of Eve. 

We translate $\M$ into a one-player $\Sigma$-labeled game with winning condition $L(\A)$, in which all states belong to Adam: we consider every state/transition of $\M$ as a vertex/edge of $\G$, take the letter output of a state as the labeling of the corresponding vertex, and ignore the other outputs and the transition labels. (See an example in Figure~\ref{fig:TransducerToGame}.)
We claim that $\A$ is not good for one-player finite-arena games, since i) Adam loses $\G$; and ii) Adam wins $\pro{\G}{\A}$.

Indeed, considering claim (i), a play of $\G$ corresponds to a possible path in $\M$, which corresponds to a possible play in Eve's letter game that agrees with $\M$. 
If there is a play of $\G$ whose labeling is not in $L(\A)$, it follows that Eve can win her letter game against $\M$, by forcing a word not in $L(\A)$, which contradicts the assumption that $\M$ is a winning strategy.

As for claim (ii), Adam can play in the $\pro{\G}{\A}$ game according to $\M$: whenever in a vertex $(v,b)$ of $\pro{\G}{\A}$, where $v$ is a vertex of $\G$ and $b$ a transition condition of $\A$, Adam chooses the next vertex according to the transition in the corresponding state in $\M$. Thus, the generated play in $\pro{\G}{\A}$ corresponds to a play in Eve's letter game that agrees with $\M$, which Adam is guaranteed to win.

In the case that $\A$ is not compliant with Adam's letter game, we do the dual: Consider the transition conditions in $\A$ to be in DNF, have a winning strategy for Eve, modeled by a transducer $\M$ whose states output a letter and a conjunctive clause and whose transitions correspond to Adam's choices, and translate it to a $\Sigma$-labeled one-player game $\G$, in which all vertices belong to Eve. Then, for analogous reasons, Eve loses $\G$, but wins $\pro{\G}{\A}$.
\end{proof}

\begin{figure}
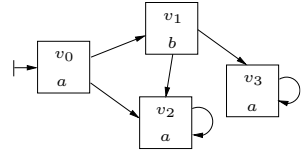 \caption{An example of a strategy for Adam in Eve's letter game, and the corresponding game, as used in the proof of Lemma~\ref{lem:GFfiniteImpliesLetter}.}\label{fig:TransducerToGame}
\end{figure}

The equivalence between the `good for trees' notion and being good for composition with one-player games, follows directly from the generalized definition of `good for trees' (Definition~\ref{def:GFTrees}) and the following observation: Every one-player $\Sigma$-labeled game is built on top of a $\Sigma$-labeled tree (its arena, in case it is infinite, or the expansion of all possible plays, in case of a finite arena), and every $\Sigma$-labeled tree can be viewed as a one-player game by assigning ownership of all positions to either Adam or Eve.
Clearly, every $\Sigma$-labeled tree $t$ belongs to $\UniDerived{L(A)}$ if and only if Eve wins the game on $t$ in which all nodes belong to Adam.

\begin{lemma}\label{lem:GoodForTrees}
An alternating automaton $\A$ is good for trees if and only if it is good for composition with one-player games.
\end{lemma}

A finite-arena game can be viewed as an alternating automaton over a singleton alphabet, suggesting that being good for composition with alternating automata implies being good for composition with finite games. This is indeed the case and by Lemmas \ref{lem:LetterImpliesGfg} and \ref{lem:GFfiniteImpliesLetter}, it also implies being good for infinite games. It turns out that even though alternating automata over a non-singleton alphabet cannot be just viewed as games, the other direction also holds.

\begin{lemma}\label{lem:AutomataComposition}
An alternating automaton $\A$ is good for game composition if and only if it is good for automata composition.
\end{lemma}

\begin{proof}
We start with showing that being good for automata composition implies being good for game composition.
Given a game over a finite $\Sigma$-arena $R=(V,X,V_E,L)$ with initial position $\iota$ and winning condition $W\subseteq \Sigma^{\omega}$, consider the automaton $A_R=(\{a\},V,\iota, \delta,\alpha)$ over the alphabet $\{a\}$ with acceptance condition $W$, where $\delta(v,a)=\bigvee \{v' | (v,v')\in X\}$ if $v\in V_E$ and $\delta(v,a)=\bigwedge \{v' | (v,v')\in X\}$ otherwise. $A_R$ accepts the unique word in $\{a\}^\omega$ if and only if Eve has a winning strategy in $R$ from $\iota$, because a strategy in $R$ is exactly a run of $A_R$ over this unique word, and it is winning if and only if the run is accepting. 

Then, observing that the synchronized product $\pro{R}{\A}$ between a finite game and an automaton is the special case of the synchronized product $\pro{A_R}{\A}$, we conclude that if an automaton is good for automata composition, then Eve wins $\pro{R}{\A}$ if and only if $\pro{A_R}{\A}$ is non-empty, if and only if $A_R$ is non-empty, i.e. if and only if Eve has a winning strategy in $R$. That is, $\A$ is good for finite game composition. From Lemmas \ref{lem:LetterImpliesGfg} and \ref{lem:GFfiniteImpliesLetter}, $\A$ is also good for composition with infinite games.\\

For the other direction, assume that $\A$ is good for game composition.
We show that $\A$ is also good for automata composition.  Consider an alternating automaton $\B$ with acceptance condition $L(\A)$. Let $w\in L(\B)$ and consider the model-checking game $\game{w}{\B}$ in which Eve has a winning strategy. Since $\A$ is good for game composition, Eve also has a winning strategy $s$ in $\pro{\game{w}{\B}}{\A}$. We use this strategy to build a strategy $s'$ for Eve in $\game{w}{\pro{\B}{\A}}$. 
 
 First recall from Def. \ref{def:prod}, that the transitions of $\pro{\B}{\A}$ are of the form 
$f(c,q)\vee f(c',q)$, $f(c,q)\wedge f(c',q)$, corresponding to choices in $\B$, 
or of the form $g(q,c)\vee g(q,c')$ or $g(q,c)\wedge g(q,c')$, corresponding to choices in $\A$.
At a position $(w,f(c,q)\vee f(c',q))$, Eve plays as $s$ plays at $((w,c\vee c'),q)$; at $(w,g(q,c)\vee g(q,c'))$ Eve plays as $s$ plays at $((w,c\vee c'),q)$. 
Since the winning condition in both games is determined by the states of $\A$ visited infinitely often, if $s$ is winning, so is $s'$. Therefore $L(\pro{\B}{\A})$ accepts $w$ and $L(\B)\subseteq L(\pro{\B}{\A})$. In the case $w \notin L(\B)$, Adam can similarly copy his strategy from $\pro{\game{w}{\B}}{\A}$ into $\game{w}{\pro{\B}{\A}}$.

We conclude that the language of $\pro{\B}{\A}$ is equal to the language of $\B$ and  therefore $\A$ is good for automata composition.
\end{proof}

\begin{remark}
We observe that compositionality with word automata also implies compositionality with (symmetric, unranked) tree automata. A tree automaton is similar to a word automaton, except that its transitions have modalities $\Box q$ and $\Diamond q$ instead of states. Then, the model-checking (or membership) game of a tree and an automaton is a game, as for words, where, in addition, the modalities $\Box q$ and $\Diamond q$ dictate whether the choice of successor in the tree is given to Adam or Eve. Then, if $\A$ composes with games, it must in particular compose with the model-checking game of $t$ and a tree automaton $\B$ with acceptance condition $L(\A)$. If Eve (Adam) wins the model-checking game $\game{t}{\B}$, she (he) also wins $\game{t}{\B}\times \A$. Her (his) winning strategy in this game is also a winning strategy in $\game{t}{\pro{\B}{\A}}$, so $\pro{\B}{\A}$ must accept (reject) $t$. $\A$ therefore composes with tree-automata.
\end{remark}

While Theorem \ref{thm:equivalence} obviously holds also for nondeterministic automata, we observe that in the absence of universality, the definitions of Section~\ref{sec:GfgDefinitions} can be relaxed into asymmetrical ones. For letter games, history-determinism, and good-for-trees, it follows directly from the definitions, as only their `nondeterministic part' applies. For composition with games and automata, we also show that it suffices to compose with universal automata and games.

\begin{lemma}\label{lem:UniversalAutomaton}
A nondeterministic automaton $\A$ is good for automata composition if and only if it is good for composition with universal automata.
\end{lemma}

\begin{proof}
Since universal automata is a subclass of alternating automata, one direction is immediate and we only need to show that if $\A$ is good for composition with all universal automata, it is good for composition with all automata.

Assume that $\A$ is good for composition with universal automata. We will show that $\A$ composes with any game $G$ with acceptance condition $L(\A)$. Assume Eve wins in $G$. Let $G'$ be the game induced by a positional winning strategy $s$ for Eve in $G$, seen as a universal automaton on the singleton alphabet. Since $\A$ composes with universal automata, it composes with $G'$, and Eve has a winning strategy $s'$ in $\pro{G'}{\A}$. Then, Eve's strategy in $\pro{G}{\A}$ consisting of using $s$ to resolve the branching in $G$ and $s'$ to resolve the nondeterminism in $\A$ is winning.
If Adam wins in $G$, then his winning strategy in $\pro{G}{\A}$ resolves the branching in $G$ according to a winning strategy. This forces the play to follow a word not in $L(\A)$. Eve has no accepting run in $\A$ for such a word and therefore can not win in $\pro{G}{\A}$ against this strategy.
\end{proof}
\section{Expressiveness}\label{sec:Expressiveness}
For some acceptance conditions, such as weak, B\"uchi, and co-B\"uchi, alternating automata are more expressive than deterministic ones. For other conditions, such as parity, Rabin, Streett, and Muller, they are not. Yet, also for the latter conditions, once considering the condition's \emph{index}, which is roughly its size, alternating automata are more expressive than deterministic automata with the same acceptance condition and index. (More details on the different acceptance conditions can be found, for example, in \cite{Bok18}.)

Most acceptance conditions are preserved, together with their index, when taking the product of an automaton $\A$ with an auxiliary memory $M$. In such a product, the states of the resulting automaton are pairs $(q,m)$ of a state $q$ from $\A$ and a state $m$ from $M$, while the acceptance condition is defined according to the projection of the states to their $\A$'s component. In particular, the weak, B\"uchi, co-B\"uchi, parity, Rabin, and Streett conditions are preserved, together with their index, under memory product, while the very-weak and Muller conditions are not.

For showing that GFG automata are not more expressive than deterministic automata with the same acceptance condition and index, we generalize the proof of \cite{BKKS13} from nondeterminism to alternation. The idea is to translate an alternating GFG automaton $\A$ to an equivalent deterministic automaton $\D$ by taking the product of $\A$ with the transducers that model the history-deterministic strategies of $\A$.

\begin{theorem}
Every alternating GFG automaton with acceptance condition that is preserved under memory-product can be translated to a deterministic automaton with the same acceptance condition and index. In particular, this is the case for weak, co-B\"uchi, B\"uchi, parity, Rabin, and Streett GFG alternating automata of any index.
\end{theorem}
\begin{proof}
Consider an alternating GFG automaton $\A=(\Sigma,Q,\iota,\delta,\alpha)$. For convenience, we may assume by Proposition~\ref{prop:TransitionConditionFormat} that $\A$'s transition conditions are in DNF.

By Corollary~\ref{cor:FiniteMemory}, the history-determinism of $\A$'s universality and nondeterminism is witnessed by finite-memory strategies, modeled by transducers $M_A=(I_A,O_A,M_A,\iota_A,\rho_A,\chi_A)$ and $M_E=(I_E,O_E,M_E,\iota_E,\rho_E,\chi_E)$, respectively. 
Observe that since transition conditions of $\A$ are in DNF, the strategy $\M_A$ chooses a state of $\A$ for every letter in $\Sigma$ and clause of states of $\A$, while the transitions in $\M_E$ are made in pairs, first choosing a clause for a letter in $\Sigma$ and then updating the memory again according to Adam's choice of a state of $\A$.

Formally, we have that the elements of $I_A$ are pairs $(a,C)$, where $a\in\Sigma$ and $C$ is a  clause of states in $Q$, and that elements of $O_A$ are states in $Q$, while elements of $I_E$ are in $\Sigma\cup Q$ and elements of $O_E$ are either clauses of states in $Q$ or $\epsilon$ (when only updating the memory).

Let $\D=(\Sigma,Q',\iota',\delta',\alpha')$ be the deterministic automaton that is the product of $\A$ and $\M_A$, in which the universality is resolved according to $\M_A$ and the nondeterminism according to $\M_E$.
That is, $Q'=Q\times M_A\times M_E$, $\iota'=(\iota,\iota_A,\iota_E)$, $\alpha'(q,x,y)=\alpha(q)$, and for every $q\in Q$, $x\in M_A$, $y\in M_E$, and $a\in\Sigma$, we have $\delta'((q,x,y),a) = (q',x',y')$, where $x' = \rho_A(x,(a, \chi_E(\rho_E(y,a))))$, $q' = \chi_A(x')$, and $y' = \rho_E(\rho_E(y,a),q')$.

Observe that $\A$ and $\D$ have the same acceptance condition and the same index, as $\A$'s acceptance condition is preserved under memory-product. We also have that $\A$ and $\D$ are equivalent, since for every word $w$, the games $\game{w}{\A}$, $\Sgame{w}{\A}{\M_A,\M_E}$, and $\game{w}{\D}$ have the same winner.
\end{proof}

\section{Succinctness}\label{sec:Succinctness}
Nondeterministic GFG automata over finite words and weak GFG automata over infinite words can be pruned to equivalent deterministic automata \cite{KSV06,Mor03}. We show that this remains true in the alternating setting.
The succinctness of nondeterministic GFG B\"uchi automata compared to deterministic ones is still an open question, having no lower bound and a quadratic upper bound, whereas nondeterministic GFG co-B\"uchi  automata can be exponentially more succinct than their deterministic counterparts \cite{KS15}. 
We show that in the alternating setting, both B\"uchi and co-B\"uchi GFG automata are singly-exponential more succinct than deterministic ones.
We leave open the question of whether stronger acceptance conditions can allow GFG automata to be doubly-exponential more succinct than deterministic ones.

In this section we focus on specific classes of automata, and for brevity use three letter acronyms in $\{$D, N, A$\} \times \{$W, B, C$\} \times \{$W$\}$ when referring to them. The first letter stands for the transition mode (deterministic, nondeterministic, alternating); the second for the acceptance-condition (weak, B\"uchi, co-B\"uchi); and the third indicates that the automaton runs on words. For example, DBW stands for a deterministic B\"uchi automaton on words. We also use DFA, NFA, and WFA when referring to automata over finite words.

In the nondeterministic setting, the proof that GFG NFAs and GFG NWWs are not more succinct than DFAs and DWWs, respectively, is based on two properties: i) In a minimal DFA or DWW for a language $L$, there is exactly one state for every Myhill-Nerode equivalence class of $L$. (Recall that finite words $u$ and $v$ are in the same class $C$ when for every word $w$, $uw\in L$ if and only if $vw\in L$. For a class $C$, the language $L(C)$ of $C$ is $\{w \St \exists u\in C \mbox{ such that } uw\in L\}$.); and ii) In a nondeterministic GFG automaton $\A$ that has no redundant transitions (removing a transition will change its language or make it not GFG), for every finite word $u$ and states $q,q'\in\delta(u)$, we have $L(\A^q) = L(\A^{q'})$.

For showing that GFG AFAs and AWWs are not more succinct than DFAs and DWWs, respectively, we provide in the following lemma a variant of the above second property.

\begin{lemma}\label{lem:MyhillNerode}
Consider a GFG alternating automaton $\A$. Then for every class $C$ of the Myhill-Nerode equivalence classes of $L(\A)$, there is a state $q$ in $\A$, such that $L(\A^q) = L(C)$.
\end{lemma}
\begin{proof}
Let $\tau$ and $\eta$ be history-deterministic strategies of $\A$ for Eve and Adam, respectively. For every finite word $u$, let $C(u)$ be the Myhill-Nerode equivalence class of $u$, and $q(u)$ be the state that $\A$ reaches when running on $u$ along $\tau$ and $\eta$. We claim that $L(\A^{q(u)}) = L(C(u))$. 

Indeed, if there is a word $w\in L(C(u)) \setminus L(\A^{q(u)})$ then Adam wins the model-checking game $\Sgame{uw}{\A}{\tau}$, by playing according to $\eta$ until reaching $q(u)$ over $u$ and then playing unrestrictedly for rejecting the $w$ suffix, contradicting the history-determinism of $\tau$. 

Analogously, if there is a word $w\in  L(\A^{q(u)}) \setminus L(C(u))$ then Eve wins the model-checking game $\Sgame{uw}{\A}{\eta}$, by playing according to $\tau$ until reaching $q(u)$ over $u$ and then playing unrestrictedly for accepting the $w$ suffix, contradicting the history-determinism of $\eta$. 
\end{proof}

The insuccinctness of GFG AFAs and GFG AWWs directly follows.

\begin{theorem}\label{thm:AfaAndAww}
For every GFG AFA or GFG AWW $\A$, there is an equivalent DFA or DWW $\A'$, respectively, such that the number of states in $\A'$ is not more than in $\A$.
\end{theorem}

\begin{proof}
The argument below corresponds to a DWW $\A$, and stands also for a DFA $\A$.
 
By \cite{Lod01}, a minimal DWW for a language $L(\A)$ has a single state for every Myhill-Nerode class of $L(\A)$. By Lemma~\ref{lem:MyhillNerode},  $\A$ has at least one state for each such class, from which the claim follows.
\end{proof}

As opposed to weak automata, minimal deterministic B\"uchi and co-B\"uchi automata do not have the Myhill-Nerode classification, and indeed, it was shown in \cite{KS15} that GFG NCWs can be exponentially more succinct than DCWs. 

We show that GFG ACWs are also only singly-exponential more succinct than DCWs. We translate a GFG ACW $\A$ to an equivalent DCW $\D$ in four steps: i) Dualize $\A$ to a GFG ABW $\B$; ii)  Translate $\B$ to an equivalent NBW, having an $O(3^n)$ state blow-up \cite{MH84,BKR10}, and prove that the translation preserves GFGness; iii) Translate $\B$ to an equivalent DBW $\C$, having an additional quadratic state blow-up \cite{KS15}; and iv) Dualize $\C$ to a DCW $\D$.

The main difficulty is, of course, in the second step, showing that the translation of an ABW to an NBW preserves GFGness. For proving it, we first need the following key lemma, stating that in a GFG ABW in which the transition conditions are given in DNF, the history-deterministic strategies can only use the current state and the prefix of the word read so far, ignoring the history of the transition conditions. 


\begin{lemma}\label{lem:GfgAbwStrategies}
Consider an ABW $\A$ with transition conditions in DNF and history-deterministic nondeterminism. Then Eve has a  strategy $\tau: \pro{\Sigma^*}{Q}\rightarrow \BP(Q)$ (and not only a strategy $(\pro{\Sigma}{\BP(Q)} )^* \rightarrow \BP(Q)$),
such that for every word $w\in L(\A)$, Eve wins $\Sgame{w}{A}{\tau}$.
\end{lemma}
\begin{proof}
Let $\xi: (\pro{\Sigma}{\BP(Q)})^* \rightarrow \BP(Q)$ be a `standard' history-deterministic strategy for Eve. 
Observe that since the transition conditions of $\A$ are in DNF, $\xi$'s domain is $(\pro{\Sigma}{Q})^*$, and the run of $\A$ on $w$ following $\xi$, namely $\Sgame{w}{\A}{\xi}$, is an infinite tree, in which every node is labeled with a state of $\A$.
A history $h$ for $\xi$ is thus a finite sequence of states and a finite word. Let $\Yearn(h)$ denote the number of positions in the sequence of states in $h$ from the last visit to $\alpha$ until the end of the sequence. We shall say ``a history $h$ for a word $u$'' when $h$'s word component is $u$, and ``$h$ ends with $q$'' when the sequence of states in $h$ ends in $q$.

We inductively construct from $\xi$ a strategy $\tau: \pro{\Sigma^*}{Q}\rightarrow \BP(Q)$, by choosing for every history $(u,q) \in \pro{\Sigma^*}{Q}$ some history $h$ of $\xi$, as detailed below, and setting $\tau(u,q) = \xi(h)$.

At start, for $u=\epsilon$, we set $\tau(\epsilon,\iota)=\xi(\epsilon,\iota)$.
In a step in which every history of $\xi$ for $u$ ends with a different state, we set $\tau(u,q) = \xi(h)$, where $h$ is the single history that ends with $q$.

The challenge is in a step in which several histories of $\xi$ for $u$ end with the same state, as $\tau$ can follow only one of them. We define that $\tau(u,q) = \xi(h)$, where $h$ is a history of $\xi$ for $u$ that ends with $q$, and $\Yearn(h)$ is maximal among the histories of $\xi$ for $u$ that ends in $q$. 
Every history of $\xi$ that is not followed by $\tau$ is considered ``stopped'', and in the next iterations of constructing $\tau$, histories of $\xi$ with stopped prefixes will not be considered.

As $\xi$ is a winning strategy for $Eve$, all paths in $\Sgame{w}{A}{\xi}$ are accepting.
Observe that $\Sgame{w}{A}{\tau}$ is a tree in which some of the paths are from $\Sgame{w}{A}{\xi}$ and some are not---whenever a history is stopped in the construction of $\tau$, a new path is created, where its prefix is of the stopped history and the continuation follows the path it was redirected to. We will show that, nevertheless, all paths in $\Sgame{w}{A}{\tau}$ are accepting.

Assume toward contradiction a path $\rho$ of $\Sgame{w}{A}{\tau}$ that is not accepting, and let $k$ be its last position in $\alpha$. The path $\rho$ must have been created by infinitely often redirecting it to different histories of $\xi$, as otherwise there would have been a rejecting path of $\xi$. Now, whenever $\rho$ was redirected, it was to a history $h$, such that $\Yearn(h)$ was maximal. Thus, in particular, this history did not visit $\alpha$ since position $k$. Therefore, by K\"onig's lemma, there is a path $\pi$ of $\Sgame{w}{A}{\xi}$ that does not visit $\alpha$ after position $k$, contradicting the assumption that all paths of $\Sgame{w}{A}{\xi}$ are accepting.
\end{proof}

We continue with showing that the translation of an ABW to an NBW preserves GFGness.

\begin{lemma}\label{lem:AbwToNbw}
Consider an ABW $\A$ for which the nondeterminism is history-deterministic. Then the nondeterminism in the NBW $\A'$ that is derived from $\A$ by the breakpoint (Miyano-Hayashi) construction is also history-deterministic.
\end{lemma}
\begin{proof}
Consider an ABW $\A=(\Sigma,Q,\iota,\delta,\alpha)$ for which the nondeterminism is history-deterministic.  We write $\overline\alpha$ for $Q\setminus\alpha$.
By Proposition~\ref{prop:TransitionConditionFormat}, we may assume that the transition conditions of $\A$ are given in DNF.

The breakpoint construction \cite{MH84} generates from $\A$ an equivalent NBW $\A'$, by intuitively maintatining a pair $\pair{S,O}$ of sets of states of $\A$, where $S$ is the set of states that are visited at the current step of a run, and $O$ are the states among them that ``owe'' a visit to $\alpha$. A state owes a visit to $\alpha$ if there is a path leading to it with no visit to $\alpha$ since the last ``breakpoint'', which is a state of $\A'$ in which $O=\emptyset$. The accepting states of $\A'$ are the breakpoints.

For providing the formal definition of $\A'$, we first construct from $\delta$ and each letter $a\in\Sigma$, a set $\Delta(a)$ of transition functions $\gamma_1, \gamma_2, \ldots \gamma_k$, for some $k\in\Nat$, such that each of them has only universality and corresponds to a possible resolution of the nondeterminism in every state of $\A$. For example, if $\A$ has states $q_1, q_2,$ and $q_3$, and its transition function $\delta$ for the letter $a$ is $\delta(q_1,a)=q_1 \land q_3 \lor q_2; \delta(q_2,a)=q_2 \lor q_3 \lor q_1$; and $\delta(q_3,a)=q_2 \land q_3 \lor q_1$, then $\Delta(a)$ is a set of twelve transition functions, where $\gamma_1(q_1,a)=q_1 \land q_3; \gamma_1(q_2,a)=q_2$; and $\gamma_1(q_3,a)=q_2 \land q_3$, etc.., corresponding to the possible ways of resolving the nondeterminism in each of the states. 

For convenience, we shall often consider a conjunctive formula over states as a set of states, for example $q_1 \land q_3$ as $\{q_1, q_3 \}$.
For a set of states $S\subseteq Q$, a letter $a$, and a transition function $\gamma\in\Delta(a)$, we define $\gamma(S) = \bigcup_{q\in S} \gamma(q,a)$.

Formally, the breakpoint construction \cite{MH84} generates from $\A$ an equivalent NBW $\A'=(\Sigma,Q',\iota',\delta',\alpha')$ as follows:
\begin{itemize}
\item $Q' = \{ \pair{S,O} \St O \subseteq S \subseteq Q \}$
\item $\iota' = (\{\iota\}, \{\iota\}\cap\overline\alpha)$
\item $\delta':$ For a state $\pair{S,O}$ of $\B$ and a letter $a\in\Sigma$:
	\begin{itemize}
	\item If $O=\emptyset$ then $\delta'(\pair{S,O},a) = \{ \pair{\hat S, \hat O} \St  \mbox{ exists a transition function } \gamma\in\Delta, \mbox{ such that } \hat S=\gamma(S,a) \mbox{ and } \hat O = \gamma(S,a) \cap \overline\alpha \}$
	\item If $O\neq\emptyset$ then $\delta'(\pair{S,O},a) = \{ \pair{\hat S, \hat O} \St  \mbox{ exists a transition function } \gamma\in\Delta, \mbox{ such that } \hat S=\gamma(S,a) \mbox{ and } \hat O = \gamma(O,a) \cap \overline\alpha \}$
	\end{itemize}	
\item $\alpha' = \{ (S, \emptyset) \St S \subseteq Q \}$
\end{itemize}

Observe that the breakpoint construction determinizes the universality of $\A$, while morally keeping its nondeterminism as is. This will allow us to show that Eve can use her history-deterministic strategy for $\A$ also for resolving the nondeterminism in $\A'$.

At this point we need Lemma~\ref{lem:GfgAbwStrategies}, guaranteeing a strategy $\tau: \pro{\Sigma^*}{Q}\rightarrow \BP(Q)$ for Eve. At each step, Eve should choose the next state in $\A'$, according to the read prefix $u$ and the current state $(S,O)$ of $\A'$. Observe that $\tau$ assigns to every state $q\in S$ a set of states $S'=\tau(u,q)$, following a nondeterministic choice of $\delta(q,u)$; Together, all these choices comprise some transition function $\gamma\in\Delta$. Thus, in resolving the nondeterminism of $\A'$, Eve's strategy $\tau'$ is to choose the transition $\gamma$ that is derived from $\tau$. 

Since $\tau$ guarantees that all the paths in the $\tau$-run-tree of $\A$ on a word $w\in L(\A)$ are accepting, the corresponding $\tau'$-run of $\A'$ on $w$ is accepting, as infinitely often all the $\A$-states within $\A'$'s states visit $\alpha$.
\end{proof}

\begin{theorem}\label{thm:GfgAbwToDbw}
The translation of a GFG ABW or GFG ACW $\A$ to an equivalent DBW or DCW, respectively, involves a $2^{\Theta(n)}$ state blow-up.
\end{theorem}

\begin{proof}
The lower bound follows from \cite{KS15}, where it is shown that determinization of GFG NCWs is in $2^{\Omega(n)}$. It directly generalizes to  GFG ACWs, and by dualization to GFG ABWs: Given a GFG ACW $\A$, we can dualize it to an ABW $\B$, which is also GFG by Proposition~\ref{prop:DualPreservesGFG}. Then, we can determinize $\B$ to a DBW $\D$ and dualize the latter to a DCW $\C$ equivalent to $\A$.

The upper bound follows from Lemma~\ref{lem:AbwToNbw}, getting an $O(3^n)$ state blow-up for translating a GFG ABW to an equivalent GFG NBW, and then another quadratic state blow-up, due to \cite{KS15}, from GFG NBW to DBW. For determinizing a GFG ACW, we have the same result due to dualization and Proposition~\ref{prop:DualPreservesGFG}.
\end{proof}
\section{Conclusions}

\Subject{Alternating GFG is the sum of its parts} 
Throughout our five definitions of GFG for alternating automata, a common theme prevails: each definition can be divided into a condition for nondeterminism and one for universality, and their conjunction guarantees good-for-gameness. For example, it suffices for an automaton to compose with both universal automata and nondeterministic automata for it to compose with alternating automata, even alternating tree automata. In other words, GFG nondeterminism and universality cannot interact pathologically to generate alternating automata that are not GFG, and neither can they ensure GFGness without each being GFG independently. This should facilitate checking GFGness, as it can be done separately for universality and nondeterminism.

\Subject{Between words, trees, games, and automata}
In the recent translations from alternating parity word automata into weak automata \cite{BL18,DJL19}, the key techniques involve adapting methods that use \textit{finite one-player games} to process \textit{infinite} structures that are in some sense between words and trees, and use these to manipulate alternating automata. These translations implicitly depend on the compositionality that enable the step from asymmetrical one-player games, i.e. trees, to alternating automata. Studying good-for-gameness provides us with new tools to move between words, trees, games, and automata, and better understand how nondeterminism, universality, and alternations interact in this context.

\bibliography{gfg}

\FullVersion{
\newpage

\appendix

\section{Additional Proofs}

}

\end{document}